\documentclass[9pt,twocolumn]{IEEEtran}
\usepackage{amsmath,amssymb,euscript ,yfonts,psfrag,latexsym,dsfont,graphicx,bbm,color,amstext,wasysym,subfig,parskip,amsbsy}
\graphicspath{{./},{./figures/}}

\begin{document}
\newtheorem{theorem}{Theorem}[section]
\newtheorem{corollary}[theorem]{Corollary}
\newtheorem{proposition}[theorem]{Proposition}
\newtheorem{lemma}[theorem]{Lemma}
\newtheorem{example}[theorem]{Example}
\newtheorem{remark}[theorem]{Remark}
\newtheorem{problem}[theorem]{Problem}
\newtheorem{definition}[theorem]{Definition}

\newcommand{\mR}{{\mathbb R}}
\newcommand{\D}{{\mathbb D}}
\newcommand{\E}{{\mathbb E}}
\newcommand{\mcN}{{\mathcal N}}
\newcommand{\mcR}{{\mathcal R}}
\newcommand{\cG}{{\mathcal G}}
\newcommand{\diag}{\operatorname{diag}}
\newcommand{\tr}{\operatorname{trace}}
\newcommand{\ignore}[1]{}
\newcommand{\mike}{\color{magenta}}
\newcommand{\bike}{\color{blue}}
\newcommand{\rike}{\color{red}}
\newcommand{\Mike}{\color{blue}}
\newcommand{\mD}{{\mathbb D}}
\newcommand{\bbD}{{\mD}}
\newcommand{\support}{{\rm Supp}}
\newcommand{\trace}{\operatorname{\rm trace}}
\renewcommand{\diag}{\mathop{\mathrm{diag}}}

\def\spacingset#1{\def\baselinestretch{#1}\small\normalsize}
\setlength{\parskip}{10pt}
\setlength{\parindent}{20pt}
\spacingset{1}

\definecolor{grey}{rgb}{0.6,0.6,0.6}
\definecolor{lgrey}{rgb}{0.9,.7,0.7}

\newcommand{\fM}{{{\mathfrak M}}}
\newcommand{\bmu}{{\pmb{\mu}}}

\title{Robust transport over networks}

\author{Yongxin Chen, Tryphon Georgiou, Michele Pavon and Allen Tannenbaum
\thanks{Y.\ Chen is with the Department of Mechanical Engineering,
University of Minnesota, Minneapolis, Minnesota MN 55455, USA; chen2468@umn.edu}
\thanks{T.T.\ Georgiou is with the Department of Electrical and Computer Engineering,
University of Minnesota, Minneapolis, Minnesota MN 55455, USA; tryphon@umn.edu}
\thanks{M.\ Pavon is with the Dipartimento di Matematica ``Tullio Levi Civita",
Universit\`a di Padova, via Trieste 63, 35121 Padova, Italy; pavon@math.unipd.it}
\thanks{A. Tannenbaum is with the Department of Computer Science, Stony Brook University, Stony Brook, NY 11794:  allen.tannenbaum@stonybrook.edu}
\thanks{Supported in part by the
NSF under Grant ECCS-1509387,
the AFOSR under Grants FA9550-12-1-0319 and FA9550-15-1-0045, the NIH under Grants P41-RR-013218, P41-EB-015902 and 1U24CA18092401A1.
This work was also supported by the Vincentine Hermes-Luh Chair, and by the University of Padova Research Project CPDA 140897.}}

\maketitle
\begin{abstract}We consider transportation over a strongly connected, directed graph.
The scheduling amounts to selecting transition probabilities for a discrete-time Markov evolution which is designed to be consistent with initial and final marginal constraints on mass transport. We address the situation where initially the mass is concentrated on certain nodes and needs to be transported in a certain time period to another set of nodes, possibly disjoint from the first.
%
The random evolution is selected to be closest to a {\em prior} measure on paths in the relative entropy sense--such a construction is known as a Schr\"odinger bridge between the two given marginals. It may be viewed as an atypical stochastic control problem where the control consists in suitably modifying the prior transition mechanism.
%
The prior can be chosen to incorporate constraints and costs for traversing specific edges of the graph, but it can also be selected to allocate {\em equal probability to all paths of equal length connecting any two nodes} (i.e., a uniform distribution on paths). This latter choice for prior transitions relies on the so-called {\em Ruelle-Bowen random walker} and gives rise to scheduling that tends to utilize all paths as uniformly as the topology allows.
Thus, this Ruelle-Bowen law ($\fM_{\rm RB}$) taken as prior, leads to a transportation plan
that tends to lessen congestion and ensures a level of robustness. 
We also show that the distribution $\fM_{\rm RB}$ on paths,
which attains the maximum entropy rate for the random walker given by the topological entropy, can itself be obtained as the time-homogeneous solution of a maximum entropy problem for measures on paths
(also a Schr\"{o}dinger bridge problem, albeit with prior that is not a probability measure). Finally we show that the paradigm of Schr\"odinger bridges as a mechanism for scheduling transport on networks can be adapted to 
graphs that are not strongly connected, as well as to weighted graphs. In the latter case, our approach may be used to design a transportation plan which effectively compromises between robustness and other criteria such as cost. Indeed, we explicitly provide a robust transportation plan which assigns {\em maximum probability} to {\em minimum cost paths} and therefore compares favourably with Optimal Mass Transportation strategies.
\end{abstract}

\section{Introduction}
Transport over networks has been the focus of a rapidly expanding literature due to its intrinsic relevance in a wide range of applications that include power transmission, traffic, financial transactions, biological systems and so on \cite{callaway2000network,watts1998collective,scott2006network,cabanes2014ants}. Furthermore, the topic relates to a host of other questions pertaining to the connectivity of graphs and the relative significance of their nodes
as in the Google PageRank problem \cite{brin2012reprint} and the study of interaction between genes in biological networks \cite{sandhu2015graph}.

Our starting point is an important insight on the relation between the topological structure of a network and the entropy rate of a random walker on the graph \cite{delvenne2011centrality,kitano2007towards}. As it turns out, there is a unique way to specify transition probabilities at each node in such a way so that all paths of equal length joining any two particular nodes have equal probability. Thereby, a measure is placed on the family
of paths between graph nodes that maximizes the entropy rate of a random walker, and this is a characteristic of the network. So far, the use of this concept has been to assign significance to each node in relation to the corresponding occupancy stationary distribution ({\em centrality measures}).

The focus in our paper is on how to schedule transportation plans across a network. The novel framework that we propose is that of the so-called Schr\"odinger bridge problem, where a flow is specified in agreement with an initial and a final marginal distribution on the nodes while, at the same time, the probability law on the paths is the closest possible to a prior in the relative entropy sense. The Ruelle-Bowen random walk provides a natural notion of ``uniform'' prior which gives equal importance to all paths. As a result, the transportation flow that is selected to agree with specified initial and final marginals tends to spread across all available paths as much as possible given the topological structure of the network. Thereby, such a flow leads to relatively low probability of conflict and congestion, and ensures a certain degree of inherent robustness of the transport plan. It is well appreciated that, typically, robustness, efficiency and cost are conflicting criteria when designing networks. 

By extending our approach to weighted graphs, we show that the choice of a prior distribution may be used to ensure that the resulting transportation attains a satisfactory compromise between robustness and other criteria such as cost. Indeed, we exhibit a {\em robust} transportation plan which assigns {\em maximum probability} to all {\em minimum cost paths}. It appears attractive when compared to Optimal Mass Transportation strategies which are not necessarily robust and where the minimum cost of transportation between any two nodes is supposed to be given.
Thus, the approach to scheduling transport based on Schr\"odinger bridges affords great flexibility. Moreover, it appears computationally attractive in view of the iterative algorithm proposed in \cite{georgiou2015positive}.

The paper is outlined as follows. In Section \ref{SBP}, we present the solution to a general Schr\"{o}dinger bridge problem (SBP), where the prior measure is not necessarily a probability measure, as a straightforward extension of the results in \cite{pavon2010discrete,georgiou2015positive}. Section \ref{THB} is devoted to solutions of the SBP with equal initial and final marginals which have a {\em time-invariant} transition mechanism so that they admit invariant measures. We establish the surprising result (Theorem \ref{uniquebridge}) that there is only one such bridge. This measure on paths can be constructed generalizing a classical result by Parry \cite{parry1964intrinsic}. In Section \ref{RB}, considering the special case of a prior transition given by the adjacency matrix, we describe the most important features of the Ruelle-Bowen random walker along the lines of \cite{delvenne2011centrality}. We observe that this measure $\fM_{\rm RB}$ on trajectories can be viewed as a solution to a ``time-homogeneous" Schr\"{o}dinger bridge problem where the prior transition mechanism is given by the adjacency matrix. Section \ref{RTG} describes our procedure to produce a robust transportation plan over a given strongly connected network: We take the Ruelle-Bowen distribution $\fM_{\rm RB}$ as prior in a Schr\"{o}dinger bridge problem with prescribed initial and final marginals. We also prove that the optimal transportation can also be obtained in one step by taking the rescaled adjacency matrix as prior transition mechanism (Proposition \ref{prop3}). In Section \ref{weighted}, we outline the extension of our approach to the cases of weighted and not strongly connected graphs. Finally, in Section \ref{examples} we illustrate our approach on a simple unweighted and weighted graph.

\section{The Discrete Schr\"odinger bridge problem}\label{SBP}

We first describe the ``ingredients" of  the discrete Schr\"odinger Bridge problem (SBP) considered in \cite{pavon2010discrete,georgiou2015positive}. In fact, we will consider a slight generalization, where the ``prior'' is not necessarily a probability law. 
The goal is to determine a time-evolution of probability distributions $\nu_t(\cdot)$ having support on a discrete space
\[\mathcal X=\{1,\ldots,n\},\] 
e.g., the nodes of a network, over a time-indexing set 
\[{\mathcal T}=\{0,1,\ldots,N\}\]
 in a way such that it matches the specified marginal distributions $\nu_0(\cdot)$ and $\nu_N(\cdot)$ and the resulting random evolution
is closest to the ``prior'' in a suitable sense. Regarding notation, we use $\mu_t(\cdot)$, $\nu_t(\cdot)$ for distributions, where typically, $\mu$ relates to a ``prior'' law while $\nu$ represents a ``new'' distribution with end-points specified and obtained by solving the SBP.

The prior law is induced by the Markovian evolution
 \begin{equation}\label{FP}
\mu_{t+1}(x_{t+1})=\sum_{x_t\in\mathcal X} \mu_t(x_t) m_{x_{t}x_{t+1}}
\end{equation}
for nonnegative distributions $\mu_t(\cdot)$ over $\mathcal X$ with $t\in{\mathcal T}$. Throughout, we assume that $m_{ij}\geq 0$ for all indices $i,j\in{\mathcal X}$ and for simplicity, for the most part, that the matrix
\[
M=\left[ m_{ij}\right]_{i,j=1}^n
\]
does not depend on $t$.
In this case, we will often assume that all entries of $M^N$ are positive.
The rows of the transition matrix $M$ do not necessarily sum up to one, in which case the ``total transported mass'' is not necessarily preserved. This is the case, in particular, of a Markov chain with  ``creation" and ``killing". 
In fact, $M$ may simply encode the topological structure of a directed network with $m_{ij}$ being zero or one, depending whether a certain transition is allowed.

The evolution \eqref{FP}, together with measure $\mu_0(\cdot)$, which we assume positive on $\mathcal X$, i.e.,
\begin{equation}\label{eq:mupositive}
\mu_0(x)>0\mbox{ for all }x\in\mathcal X,
\end{equation}
 induces
a measure $\fM$ on $\mathcal X^{N+1}$ as follows. It assigns to a path  $x=(x_0,x_1,\ldots,x_N)\in\mathcal X^{N+1}$ the value
\begin{equation}\label{prior}\fM(x_0,x_1,\ldots,x_N)=\mu_0(x_0)m_{x_0x_1}\cdots m_{x_{N-1}x_N},
\end{equation}
and gives rise to a flow
of {\em one-time marginals}
\[\mu_t(x_t) = \sum_{x_{\ell\neq t}}\fM(x_0,x_1,\ldots,x_N), \quad t\in\mathcal T.\]
 The ``prior" distribution $\fM$ on the space of paths may be at odds with a pair of specified marginals $\nu_0$ and $\nu_N$ in that one or possibly both,
\[
\mu_0(x_0)\neq \nu_0(x_0),
~~~\mu_N (x_N)\neq \nu_N (x_N).
\]

We denote by ${\mathcal P}(\nu_0,\nu_N)$ the family of probability distributions on $\mathcal X^{N+1}$ having the prescribed marginals. We seek a distribution in this set which is closest to the prior $\fM$ in a suitable entropic sense. To this end, let us first recall the definition of relative entropy for probability distributions.

\begin{definition}The {\em Relative Entropy} between the probability distributions $P$ and $Q$ is
\begin{equation*}
\bbD(P\|Q):=\left\{\begin{array}{ll} \sum_{x}P(x)\log\frac{P(x)}{Q(x)}, & \support (P)\subseteq \support (Q),\\
+\infty , & \support (P)\not\subseteq \support (Q),\end{array}\right.
\end{equation*} 
where, by definition,  $0\cdot\log 0=0$ and the summation is over the common set where they are supported.
\end{definition}

As is well known, $\bbD(P\|Q)$ is not symmetric and does not satisfy the triangle inequality. It does, however, satisfy $\bbD(P\|Q)\ge 0$ and $\bbD(P\|Q)=0$ if and only if $Q=P$, see, e.g., \cite{cover2012elements}. It can also be extended to positive measures that are not  probability distributions. In fact, it is quite common to consider $Q$ to be a uniform measure that may not be a probability measure, such as the Lebesgue measure or the stationary Wiener measure \cite{Fol88}. Naturally, while the value of $\bbD(P\|Q)$ may turn out negative due to miss-match of scaling, the relative entropy is always jointly convex.
We view the prior $\fM$ (specified by $M$ and $\mu_0$) in a similar manner, and consider the Schr\"odinger Bridge  problem:

\begin{problem}\label{prob:optimization}
Determine
 \begin{eqnarray}\label{eq:optimization}
\fM^*[\nu_0,\nu_N]={\rm argmin}\{ \bbD(P\|\fM) \mid  P\in {\mathcal P}(\nu_0,\nu_N)
\}.
\end{eqnarray}
\end{problem}

Provided all entries of $M^N$ are positive,
the problem has a solution, which is unique due to strict convexity.
This is stated next.

\begin{theorem}\label{solbridge} Assume that $M^N$ has all positive elements.
There exist nonnegative functions  $\varphi(\cdot)$ and $\hat{\varphi}(\cdot)$ on $[0,N]\times\mathcal X$ satisfying for $t\in[0,N-1]$ the system
\begin{subequations}\label{eq:Schroedingersystem}
\begin{eqnarray}\label{Schroedingersystem1}
\varphi(t,i)&=&\sum_{j}m_{ij}\varphi(t+1,j),\\\hat{\varphi}(t+1,j)&=&\sum_{i}m_{ij}\hat{\varphi}(t,i),\label{Schroedingersystem2}
\end{eqnarray}
with the boundary conditions 
\begin{eqnarray}\label{bndconditions1}
\varphi(0,x_0)\cdot\hat{\varphi}(0,x_0)&=&\nu_0(x_0)\\\label{bndconditions2}
\varphi(N,x_N)\cdot\hat{\varphi}(N,x_N)&=&\nu_N(x_N),
\end{eqnarray}
\end{subequations}
for all $x_0, x_N\in\mathcal X$.
Moreover, the solution $\fM^*[\nu_0,\nu_N]$ to Problem \ref{prob:optimization} is unique and obtained by
\[
\fM^*[\nu_0,\nu_N](x_0,\ldots,x_N)=\nu_0(x_0)\pi_{x_0x_{1}}(0)\cdots \pi_{x_{N-1}x_{N}}(N-1),
\]
where\footnote{Here we use the convention that $0/0=0$.}
\begin{equation}\label{opttransition1}\pi_{ij}(t):=m_{ij}\frac{\varphi(t+1,j)}{\varphi(t,i)}.
\end{equation}
Equation \eqref{opttransition1} specifies {\em one-step transition probabilities} that are well defined. 
\end{theorem}

\begin{proof}The argument in \cite[Theorem 4.1]{pavon2010discrete} and \cite[Section III]{georgiou2015positive} applies verbatim to this setting which is slightly more general in that $M$ does not prescribe a probability kernel. The system (\ref{Schroedingersystem1}-\ref{bndconditions2}) is known as a Schr\"odinger system. The existence of solution is shown in \cite[Section III]{georgiou2015positive} by establishing that
the composition
\begin{align}\label{eq:composition}\nonumber
\hat\varphi(0,x_0)  &\overset{ (M^T)^N}{\longrightarrow} \hat\varphi(N,x_N)\overset{\eqref{bndconditions2}}{\longrightarrow}\varphi(N,x_N)
\longrightarrow\dots\\[.1in] &\ldots\overset{M^N}{\longrightarrow} \varphi(0,x_0)
\overset{\eqref{bndconditions1}}{\longrightarrow}\left(\hat\varphi(0,x_0)\right)_{\rm next}
\end{align}
is contractive in the Hilbert metric \cite{lemmens2012nonlinear,birkhoff1957extensions,bushell1973projective,bushell1973hilbert}.
The fact that $\pi_{ij}(t)$ in \eqref{opttransition1} satisfy $\sum_{j}\pi_{ij}(t)=1$ follows from (\ref{Schroedingersystem1}).
\end{proof}

Notice that $\varphi$ and $\hat{\varphi}$ are unique up to multiplication of $\varphi$ by a positive constant and division of $\hat{\varphi}$ by the same constant. This is a direct consequence of the proof above as the Hilbert metric is a metric on the projective space. The statement of the theorem is analogous to results for the classical Schr\"{o}dinger system \eqref{eq:Schroedingersystem} of diffusions that have been established by Fortet, Beurling, Jamison and F\"{o}llmer \cite{For40,Beu60,Jam74,Fol88}. The requirement that $M^N$ has only positive entries can be slightly relaxed and replaced by a suitable condition that guarantees existence of solution for the particular $\nu_0$ and $\nu_N$. The case when $M$ is time-varying can also be readily established along the lines of \cite[Theorem 4.1]{pavon2010discrete} and \cite[Theorem 2]{georgiou2015positive}.

Finally, to simplify the notation, let $\varphi(t)$ and $\hat{\varphi}(t)$ denote the column vectors with components $\varphi(t,i)$ and $\hat{\varphi}(t,i)$, respectively, with $i\in\mathcal X$. In matricial form, (\ref{Schroedingersystem1}), \eqref{Schroedingersystem2} and (\ref{opttransition1}) read
\begin{subequations}\label{eq:notations}
\begin{equation}
\varphi(t)=M\varphi(t+1),\; ~~\hat{\varphi}(t+1)=M^T\hat{\varphi}(t),
\end{equation}
and
\begin{equation}\label{matrixtransition}
	\Pi(t)=[\pi_{ij}(t)]=\diag(\varphi(t))^{-1}M\diag(\varphi(t+1)).
\end{equation}
\end{subequations}

\section{Time-homogeneous bridges}\label{THB}
In this section, we consider the case of Schr\"odinger bridge problems when the marginals are identical, namely, $\nu_0=\nu_N=\nu$. In particular, we are interested in the case when the solution of the SBP corresponds to a time-homogeneous Markov evolution. Note that, from Theorem \ref{solbridge}, $\fM^*[\nu,\nu]$ is in general time inhomogeneous.
 We first recall the following celebrated result on the spectral properties of nonnegative matrices \cite{horn2012matrix}.

\begin{theorem}[Perron-Frobenius]\label{perrontheorem}
Let $A=\left(a_{ij}\right)$ be an $n\times n$ matrix with nonnegative elements. Suppose there exists $N$ such that $A^N$ has only positive elements, 
and let $\lambda_A$ be its spectral radius. Then\\[-.25in]
\begin{enumerate}
\item[i)] $\lambda_A>0$ is an eigenvalue of $A$;
\item[ii)] $\lambda_A$ is a simple eigenvalue;
\item[iii)] there exists an eigenvector $v$ corresponding to $\lambda_A$ with strictly positive entries;
\item[iv)] $v$ is the only non-negative eigenvector of A;
\item[v)] \label{item_v} let $B=[b_{ij}]$ be an $n\times n$ matrix with nonnegative elements. If $a_{ij}\le b_{ij}$, $\forall i,j\leq n$ and $A\neq B$, then $\lambda_A<\lambda_B$. 
\end{enumerate}
\end{theorem}


Since the nonnegative matrix $M$ is such that $M^N$ has only positive elements, by the above Perron-Frobenius Theorem, $M$ has a unique positive eigenvalue $\lambda_M$ which is equal to the spectral radius. Let $\phi$ and $\hat{\phi}$ be the corresponding right and left eigenvectors and denote their entries by $\phi(x)$ and $\hat\phi(x)$ with $x\in\mathcal X$, respectively. Then both 
have only positive entries. We normalize $\phi$ and $\hat\phi$ so that
\[
\sum_{x\in\mathcal X}\phi(x)\hat{\phi}(x)=1.
\]
This leads to a special probability distribution
\begin{equation}\label{barnu}
\bar{\nu}(x)=\phi(x)\hat{\phi}(x).
\end{equation}
It turns out that $\bar{\nu}$ is the only probability measure such that the associated SBP has a time-homogeneous solution; we shall name it the {\em time-homogeneous bridge} associated with $M$. It admits the following variational characterization.

\begin{proposition} \label{prop} 
Let $M$ be a nonnegative matrix such that $M^N$ has only positive elements, and $\fM$ the measure on $\mathcal X^{N+1}$ given by \eqref{prior} with $\mu_0$ satisfying \eqref{eq:mupositive}. Then the solution to the Schr\"odinger bridge problem
 \begin{equation}\label{eq:optimization_stat}
\fM^*[\bar{\nu},\bar{\nu}]={\rm argmin}\{ \bbD(P\|\fM) | P\in {\mathcal P}(\bar{\nu},\bar{\nu})
\},
\end{equation}
where $\bar{\nu}$ is as in (\ref{barnu}), has the time-invariant transition matrix
\begin{equation}\label{stationarymatrixopttransition}
\bar{\Pi}=\lambda_M^{-1}\diag(\phi)^{-1}M\diag(\phi)
\end{equation}
and invariant measure $\bar{\nu}$. 
\end{proposition}
\begin{proof}
Since $\phi$ and $\hat\phi$ are the right and left eigenvectors of $M$ associated with eigenvalue $\lambda_M$, the nonnegative functions $\varphi$ and $\hat\varphi$ defined by
    \[
         \varphi(t,x)=\lambda_M^t \phi(x), ~~
        \hat\varphi(t,x)=\lambda_M^{-t} \hat\phi(x)
    \]
satisfy the Schr\"odinger system \eqref{eq:Schroedingersystem}.
By Theorem \ref{solbridge}, the solution $\fM^*[\bar{\nu},\bar{\nu}]$ of the Schr\"odinger bridge problem \eqref{eq:optimization_stat} then has the transition matrix (see (\ref{matrixtransition}))
\begin{eqnarray*}
    \bar{\Pi}&=&\diag(\varphi(0))^{-1}M\diag(\varphi(1))
    \\&=&\lambda_M^{-1}\diag(\phi)^{-1}M\diag(\phi),
\end{eqnarray*}
which is exactly \eqref{stationarymatrixopttransition}. Moreover, since
\[
\bar{\Pi}^T\bar{\nu}=\lambda_M^{-1}\diag(\phi)M^T\hat{\phi}=\bar{\nu},
\]
it follows that $\bar{\nu}$ is the corresponding invariant measure.
\end{proof}
In particular, notice that $\fM^*[\bar{\nu},\bar{\nu}]$, and its extension to infinite paths $x=(x_0, x_1,x_2,\ldots)$ through (\ref{stationarymatrixopttransition}), is {\em stationary}.  Indeed, we have the following more general result which is of independent interest.
\begin{proposition} \label{stationarychain}Let $P\in {\mathcal P}(\nu,\nu)$ be a Markovian measure on $\mathcal X^{N+1}$ having time-invariant transition matrix $\Pi$. Then $\nu$ is invariant for $\Pi$, i.e. $\Pi^T\nu=\nu$.
\end{proposition}
\begin{proof}Let $\Pi^T\nu=m$. Then 
\[
d_H(\nu,m)=d_H((\Pi^T)^N\nu,(\Pi^T)^Nm)\le\lambda d_H(\nu,m)
\]
where $d_H$ is the Hilbert distance \cite{birkhoff1957extensions,georgiou2015positive} and $\lambda<1$ is the contraction ratio of the map $(\Pi_\nu^T)^N$. Since both $\nu$ and $m$ are probability distributions, it follows that $m=\nu$ and $\nu$ is invariant.
\end{proof}
We show next  that, under mild assumptions, there is only one time-homogeneous bridge between equal marginals. In the following result, we shall use the following notation. As before, let $\fM$ be given by \eqref{prior} with $\mu_0$ satisfying \eqref{eq:mupositive}. We denote by $\fM(2N)$ the unique extension of $\fM$  to all of the discrete interval $[0,2N]$ by its time-invariant transition mechanism. We also denote by $\fM^*[\nu,\nu](2N)$ the Schr\"odinger bridge with prior $\fM(2N)$ and equal marginals $\nu$ at times $t=0$ and $t=2N$.
\begin{theorem}\label{uniquebridge}
Let $M$ be a nonnegative matrix such that $M^N$ has only positive elements. Let $\nu$ be a probability measure. Suppose $N>1$ and that the transition matrix $\Pi_\nu$ of $\fM^*[\nu,\nu](2N)$ does not depend on time. Then $\nu=\bar{\nu}$ and $\Pi_\nu=\bar{\Pi}$.
\end{theorem}
\begin{proof}
 Let $\varphi_{\nu}(t)=M\varphi_{\nu}(t+1)$ be the space-time harmonic function associated to the minimizer $\fM^*[\nu,\nu]$.  Suppose first that $M$ has only positive elements and consider  times $t=N-2, N-1, N$. By (\ref{eq:notations}) and the time invariance of $\Pi_\nu$, we must have 
\begin{eqnarray}\nonumber\Pi_\nu=\diag(\varphi_\nu(N-2))^{-1}M\diag(\varphi_\nu(N-1))\\=\diag(\varphi_\nu(N-1))^{-1}M\diag(\varphi_\nu(N)).\nonumber
\end{eqnarray}
It follows that 
\[
M=D_\nu(N-1)MD_\nu(N)^{-1}, 
\]
where 
\begin{align*}
D_\nu(t)&=\diag(\varphi_\nu(t))\diag(\varphi_\nu(t-1))^{-1}\\
&=\diag (d_1^\nu(t),\ldots, d_n^\nu(t))
\end{align*}
is diagonal for all $t$.
Hence,
\[m_{ij}=e_i^TMe_j=d_i(N-1)m_{ij}d_j(N)^{-1}, \quad \forall i,j.
\]
Varying $j$ for a fixed $i$, since $m_{ij}\neq 0$, we get that $D(N)$ is a scalar matrix, say $\lambda I$, not dependent on $t$ and $\varphi(N)$ is a right eigenvector of $M$.  By the Perron-Frobenius Theorem, it follows that $\varphi(N)$ corresponds to $\lambda_M$. It readily follows that $\hat{\varphi}(0)$ is an eigenvector of $M^T$ with positive components corresponding to the same eigenvalue $\lambda_M$. By (\ref{bndconditions1})-(\ref{bndconditions2}), $\nu$ is equal to $\bar\nu$. 

A similar argument establishes the result when $M$ has merely nonnegative elements. Indeed, looking at the $N$-step transition matrix $\Pi_\nu^{(N)}=\Pi_\nu^N$ on the time intervals $[0,N]$ and $[N,2N]$ the same argument as in the full positive case gives that $\varphi(2N)$ is a right eigenvector of $M^N$ with positive components. But so is $\varphi$. By Theorem \ref{perrontheorem}, iv), they can be taken to be equal.
\end{proof}

Consider now the following special case. We have a strongly connected, aperiodic directed graph $(\mathcal V,\mathcal E)$ with vertex set $\mathcal V=\{1,2,\ldots,n\}$ and edge set $\mathcal E\subseteq \mathcal V\times\mathcal V$. Let $A$ be the adjacency matrix of the graph so that $a_{ij}=1$ if there is an edge from $i$ to $j$ and $a_{ij}=0$ otherwise. Then, there exists $N$ such that $A^N$ has all positive elements. As we shall see in the next section, the Schr\"{o}dinger bridge problem (\ref{eq:optimization_stat}) just considered with $M=A$ as prior transition turns out to have as solution the Ruelle-Bowen measure $\fM_{\rm RB}$  \cite[Section III]{delvenne2011centrality}. This probability measure has a number of useful properties, in particular it  gives the same probability to paths of the same length between any two given nodes.  All of this is discussed in the next section.

\section{The Ruelle-Bowens random walk}\label{RB}
In this section, we follow closely the beautiful paper \cite{delvenne2011centrality} by Delvenne and Libert, which explains the Ruelle-Bowens (RB) random walk. The RB random walk amounts to a Markovian evolution on a directed graph that assigns equal probabilities to all paths of equal length between any two nodes. The motivation of \cite{delvenne2011centrality} was to assign a natural invariant probability to nodes based on relations that are encoded by a graph, and thereby determine a {\em centrality measure}, akin to Google Page ranking, yet more robust and discriminating. Our motivation is quite different. The RB random walk provides a uniform distribution on paths. Therefore, it represents a natural distribution to serve as prior in the SBP in order to achieve a maximum spreading of the mass transported over the available paths. 
In this section, besides reviewing basics on the RB random walk, we show that the RB distribution is itself a solution to the Schr\"odinger bridge problem \ref{prob:optimization}.

We consider a strongly connected, directed graph 
\[\mathcal G=(\mathcal V,\mathcal E).\]
The idea in Google Page ranking the nodes is based on a random walk where a jump takes place from one node to any of its neighbors with equal probability. The alternative proposed in \cite{delvenne2011centrality} is an {\em entropy ranking}, based on the stationary distribution of the RB random walk \cite{parry1964intrinsic,ruelle2004thermodynamic}.
The transition mechanism is such that it induces a uniform distribution on {\em paths} of equal length joining any two nodes. This distribution is characterized as the one maximizing the {\em entropy rate}  \cite{cover2012elements} for the random walker.
Let us briefly recall the relevant concept. The Shannon entropy for paths of length $t$ is at most
\[
\log |\{{\rm paths\; of \;length} \;t\}|.
\]
Hence, the entropy rate is bounded by the {\em topological entropy rate}
\[
H_{\mathcal G}=\limsup_{t\rightarrow\infty}[\log |\{{\rm paths\; of \;length} \;t\}|/t].
\]
Here $|\{\cdot\}|$ denotes the cardinality of a set.
Notice that $H_{\mathcal G}$ only depends on the graph $\mathcal G$ and not on the probability distribution on paths. More specifically, if $A$ denotes the adjacency matrix of the graph, the number of paths of length $t$ is the sum of all the entries of $A^t$. Thus, it follows that $H_{\mathcal G}$ is the logarithm of the spectral radius of $A$, namely the maximum of the absolute values of the eigenvalues of $A$,
that is
\begin{align} H_{\mathcal G}= \log(\lambda_A). \label{eq:loglambda}
\end{align}

We next construct the Rulle-Bowen random walk. Let   $A$ as in the Perron-Frobenius Theorem \ref{perrontheorem} and let $u$ and $v$ be its left and right eigenvectors\footnote{We are now conforming to notation in \cite{delvenne2011centrality} for ease of comparison. Hence we use $u$ and $v$ rather than $\hat{\phi}$ and $\phi$.} with positive components corresponding to $\lambda_A$, so that
\begin{equation}\label{eq:adjacency}A^Tu=\lambda_A u, \quad Av=\lambda_A v.
\end{equation}
Suppose $u$ and $v$ are chosen so that
\[
\langle u,v\rangle:=\sum_i u_iv_i=1.
\]
As in the previous section, it is readily seen that their componentwise multiplication
\begin{equation}\label{invariantmeasure1}
\nu_{RB}(i)=u_iv_i  
\end{equation}
defines a probability distribution which is invariant under the transition matrix
\begin{equation} R=[r_{ij}], \quad r_{ij}=\frac{v_j}{\lambda_A v_i}a_{ij}.\label{Opttransition1}  
\end{equation}
that is,
\begin{equation}\label{optevol1}
R^T\nu_{RB}=\nu_{RB}.
\end{equation}
If $A$ in \eqref{eq:adjacency} represents
the adjacency matrix $A$ of a graph, then
the transition matrix $R$ in (\ref{Opttransition1}) together with the stationary measure $\nu_{RB}$ in (\ref{invariantmeasure1}), define the {\em Ruelle-Bowen path measure}
\begin{equation}\label{uniformmeasure}
\fM_{\rm RB}(x_0,x_1,\ldots,x_N):=\nu_{RB}(x_0)r_{x_0x_1}\cdots r_{x_{N-1}x_N}.
\end{equation}
\begin{proposition} \label{uniformpath}The measure $\fM_{\rm RB}$ {\em (\ref{uniformmeasure})} assigns probability $\lambda_A^{-t}u_iv_j$ to any path of length $t$ from node $i$ to node $j$.
\end{proposition}

\begin{proof} Starting from the stationary distribution \eqref{invariantmeasure1}, and in view of \eqref{Opttransition1}, the probability of a path $ij$ is
\[
u_iv_i \left(\frac{1}{\lambda_A} v_i^{-1}v_j\right)=\frac{1}{\lambda_A}u_iv_j,
\]
assuming that node $j$ is accessible from node $i$ in one step.
Likewise, the probability of the path $ijk$ is
\[
u_iv_i \left(\frac{1}{\lambda_A} v_i^{-1}v_j\right)\left(\frac{1}{\lambda_A} v_j^{-1}v_k\right)=\frac{1}{\lambda_A^2}u_iv_k
\]
independent of the intermediate state $j$, and so on.
Thus, the claim follows.
\end{proof}


The striking property of $\fM_{\rm RB}$ is that induces a uniform probability measure on paths of equal length between any two given nodes. We quote from \cite{delvenne2011centrality} ``Since the number of paths of length $t$ is of the order of $\lambda_A^t$ (up to a factor) the distribution on paths of fixed length is uniform up to a factor (which does not depend on $t$). Hence the Shannon entropy of paths of length $t$ grows as $t\log\lambda_A$, up to an additive constant. The entropy rate of this distribution is thus $\log\lambda_A$ which is optimal'' by the expression for $H_{\mathcal G}$ in \eqref{eq:loglambda}.

The analysis also shows  that the Ruelle-Bowen distribution is the solution of the particular SBP where the ``prior" transition mechanism is given by the adjacency matrix! This observation is apparently new and beautifully links the topological entropy rate to a maximum entropy problem on path space.
We state next this special case of Proposition \ref{prop}.

\begin{proposition}\label{prop2}
Let $A$ be the adjacency matrix of a strongly connected  aperiodic graph $\mathcal G$. Let $\fM$ the nonnegative measure on $\mathcal X^{N+1}$ given by \eqref{prior} with $M=A$ and $\mu_0$ satisfying \eqref{eq:mupositive}. Then, the Ruelle-Bowen measure $\fM_{\rm RB}$ \eqref{uniformmeasure} solves the SBP (\ref{eq:optimization_stat}) with marginals $\nu_0=\nu_N=\nu_{RB}$.
\end{proposition}

\section{Robust transport over networks}\label{RTG}
Once again we consider a strongly connected, directed graph $\mathcal G=(\mathcal V,\mathcal E)$ with $n$ vertices. We identify node $1$ as a {\em source} and node $n$ as a {\em sink} and seek to transport a unit mass from $1$ to $n$ in at most $N$ steps. The task is formalized by setting an initial marginal distribution $\nu_0(x)=\delta_{1x}(x)$ Kronecker's delta. Similarly, the final distribution is $\nu_N(x)=\delta_{nx}(x)$. Generally, we seek a transportation plan which is {\em robust} and avoids {\em congestion} as much as the topology of the graph permits. This latter feature of the transportation plan will be achieved in this section indirectly, without explicitly bringing into the picture the capacity of each edge (this is done in Section \ref{weighted}). With these two key specifications in mind, we like to control the flux so that the initial mass {\em spreads as much as possible} on the feasible paths joining vertices $1$ and $n$ in $N$ steps before reconvening at time $N$ in vertex $n$. We shall achieve this by constructing a suitable Markovian transition mechanism. As we want to allow for the possibility that all or part of the mass reaches node $n$ at some time less than $N$, we always include a loop in node $n$ so that our adjacency matrix $A$ always has $a_{nn}=1$. We observed in the previous section that the Ruelle-Bowen  $\fM_{\rm RB}$ measure on paths can be obtained as the solution of the maximum entropy problem when the ``prior transition matrix" is the adjacency matrix. Since $\fM_{\rm RB}$ gives equal probability to paths joining two specific vertices, it is natural to use it as a prior in a new maximum entropy problem with marginals $\delta_{1x},\delta_{nx}$ so as to achieve the spreading of the probability mass on the feasible paths joining the source with the sink. Thus, we consider the following maximum entropy problem
\begin{problem}\label{prob:secondmaxent}
Determine
 \begin{equation}\label{eq:secondmaxent}
\fM^*[\delta_{1x},\delta_{nx}]={\rm argmin}\{ \bbD(P\|\fM_{\rm RB}) |  P\in {\mathcal P}(\delta_{1x},\delta_{nx})
\}.\nonumber
\end{equation}
\end{problem}

By Theorem \ref{solbridge}, the optimal, time varying transition matrix $\Pi^*(t)$ of the above problem is given, recalling the notations in \eqref{eq:notations}, by
\begin{equation}\label{matrixopttransition}
\Pi^*(t)=\diag(\varphi(t))^{-1}R \diag(\varphi(t+1)),
\end{equation}
where
\[
	\varphi(t)=R\varphi(t+1), ~\hat{\varphi}(t+1)=R^T\hat{\varphi}(t),
\]
with the boundary conditions
\begin{equation}\label{secondbndconditions}
\varphi(0,x)\hat{\varphi}(0,x)=\delta_{1x}(x),\;\varphi(N,x)\hat{\varphi}(N,x)=\delta_{nx}(x)
\end{equation}
for all $x\in\mathcal X$. In view of (\ref{Opttransition1}), if we define
\[
\varphi_v(t):=\lambda_A^{-t}\diag(v)\varphi(t), \quad \hat{\varphi}_v(t):=\lambda_A^t \diag(v)^{-1}\hat{\varphi}(t),
\]
then we have
\[
\varphi_v(t)=A\varphi_v(t+1), ~ \hat{\varphi}_v(t+1)=A^T\hat{\varphi}_v(t),~t=0,\ldots,N-1.
\]
Moreover,
\[
\varphi_v(t,x)\hat{\varphi}_v(t,x)=\varphi(t,x)\hat{\varphi}(t,x),~ t=0,\ldots,N-1,~ x\in\mathcal X.
\]
Here, again, $A$ is the adjacency matrix of $\cG$ and $v$ is the right eigenvector corresponding to the spectral radius $\lambda_A$.

The above analysis provides another interesting way to express $\fM^*[\delta_{1x},\delta_{nx}]$; it also solves the Schr\"odinger bridge problem with the same marginals $\delta_{1x}$ and $\delta_{nx}$ while different prior transition matrix $A$, the adjacency matrix. Thus, we can replace the two-step procedure by a single bridge problem. This is summarized in the following proposition.
\begin{proposition}\label{prop3}
Let $A$ be the adjacency matrix of a strongly connected aperiodic graph $\mathcal G$, $\fM$ the nonnegative measure on $\mathcal X^{N+1}$ given by \eqref{prior} with $M=A$ and $\mu_0$  satisfying \eqref{eq:mupositive}, then, the solution $\fM^*[\delta_{1x},\delta_{nx}]$ of Problem \ref{prob:secondmaxent} also solves the Schr\"{o}dinger bridge problem
 \begin{equation}\label{onestepSBP}
{\rm min}\{ \bbD(P\|\fM) |  P\in {\mathcal P}(\delta_{1x},\delta_{nx})
\}.
\end{equation}
\end{proposition}

The iterative algorithm of \cite[Section III]{georgiou2015positive} can now be based on (\ref{onestepSBP}) to efficiently compute the transition matrix of the optimal robust transport plan $\fM^*[\delta_{1x},\delta_{nx}]$.

\begin{remark}
Finally, observing that if $A^N$ has also zero elements, the robust transport described in this section may still be feasible provided there is at least one path of length $N$ joining node $1$ with node $n$, i.e., $(A^N)_{1n}>0$.
\end{remark}

As we discussed in the beginning of this section, the intuition to use $\fM_{RB}$ as a prior is to achieve the spreading of the probability on all the feasible paths connecting the source and the sink. It turns out this is in deeded the case; the solution $\fM^*[\delta_{1x},\delta_{nx}]$ of Problem \ref{prob:secondmaxent} assigns equal probability to all the feasible paths of lengths $N$ joining the source $1$ with the sink $n$. Too see this, by \eqref{matrixopttransition}, the probability of the optimal transport plan $\fM^*[\delta_{1x},\delta_{nx}]$ assigns on path $x=(x_0,\,x_1,\ldots,x_N)$ is
\begin{eqnarray*}
	\fM^*[\delta_{1x},\delta_{nx}](x)\!\!&=&\!\!\delta_{1x}(x_0)\!\!\prod_{t=1}^{N-1} r_{x_t x_{t+1}}\!\!\frac{\varphi(t+1,x_{t+1})}{\varphi(t,x_t)}
	\\\!\!&=&\!\!\delta_{1x}(x_0)\frac{\varphi_v(N,x_N)}{\varphi_v(0,x_0)}\prod_{t=1}^{N-1} a_{x_t x_{t+1}}.	
\end{eqnarray*}
Observing that $\prod_{t=1}^{N-1} a_{x_t x_{t+1}}=1$ for feasible path and $0$ otherwise, and $\delta_{1x}(x_0)\varphi_v(N,x_N)/\varphi_v(0,x_0)$ depends only on the boundary points $x_0, x_N$, we conclude that $\fM^*[\delta_{1x},\delta_{nx}]$ assigns equal probability to all the feasible paths. Moreover, there are $(A^N)_{1n}$ feasible paths of length $N$ connecting nodes $1$ and $n$. Thus we establish the following.
\begin{proposition}\label{uniformfeasiblepaths}
$\fM^*[\delta_{1x},\delta_{nx}]$ assigns probability $1/(A^N)_{1n}$ to each of all the feasible paths of length $N$ connecting $1$ and $n$.
\end{proposition}
 
\section{Generalization: Not strongly connected and weighted graphs} \label{weighted}
Consider again a directed graph $\mathcal G=(\mathcal V,\mathcal E)$ with $n$ vertices. We associate to the edge $ij$ an ``energy" $U_{ij}\ge 0$. We study the following two specific cases (and their combination):
 
\noindent a) Graphs that are not strongly connected:
We consider the same problem as in the previous section but the graph is not strongly connected. Following \cite{delvenne2011centrality}, we can give a large positive energy $U_0$ to non existing links (this kind of ``teleportation" is employed in the random walk of the Google Page rank algorithm to avoid getting stuck in  absorbing states) and energy $U_{ij}=0$ to existing links. Then the adjacency matrix $A$ is replaced by the matrix 
 \[
 	B=[b_{ij}]= \left[\exp(-U_{ij})\right].
\]
  The matrix $B$ has all positive elements. Hence, we can apply the Perron-Forbenius theorem.  Let $u$ and $v$ be left and right eigenvectors with positive components of the matrix $B$ corresponding to the spectral radius $\lambda_B$ of $B$, so that
$$B^Tu=\lambda_B u, \quad Bv=\lambda_B v.
$$ 
Suppose that $u$ and $v$ are chosen so that $\langle u,v\rangle=\sum_i u_iv_i=1$. Then $\mu_U$ given by
\begin{equation}\label{invariantmeasure}
\mu_U(i)=u_i\cdot v_i
\end{equation}
is a probability distribution which is invariant for the transition matrix
\begin{equation}\label{opttransition}R_U=\lambda_B^{-1}\diag(v)^{-1}B\diag(v), 
\end{equation}
namely
\begin{equation}\label{optevol}
R_U^T\mu_U=\mu_U.
\end{equation}
The corresponding path space measure $\fM_U$ is no longer uniform on paths of equal length. Indeed, the probability of the path $(i=x_0,x_1,\ldots, x_{t-1},j=x_t)$ is 
\[
	\lambda_B^{-t}\exp (-\sum_{\ell=0}^{t-1}U_{x_\ell x_{\ell+1}})u_iv_j.
\]  
However, it is the minimum {\em free energy rate}  (topological pressure in thermodynamics) distribution attaining the maximum value of $-F=-\bar{U}+S$
given by $\log\lambda_B$ and has therefore the form of a {\em Boltzmann distribution}, see \cite[Section IV]{delvenne2011centrality} for details. Notice that, as soon as there are virtual links, $B\neq A$. By statement v) in Theorem \ref{perrontheorem}, we then have $\log\lambda_A<\log\lambda_B$. Namely, the topological entropy has increased in accordance to our intuition. The expected total path energy of a path of length $t$ is precisely $t\cdot\bar{U}$. \\

Again, as in Proposition \ref{prop2}, we have a special case of Proposition \ref{prop}. Namely, the measure $\fM_U$ is the solution of  a SBP where the prior $\fM$ is a Markovian measure on $\mathcal X^{N+1}$ as in (\ref{prior}) but with transition mechanism given by $M=B$. If $U_0$ is very large, most of the transportation will occur on the real edges.  We can then take $\fM_U$  as the prior distribution in a maximum entropy problem as in Section \ref{RTG} obtaining again through the solution $\fM^*_U[\delta_{1x},\delta_{nx}]$ a robust transportation plan from node $1$ to node $n$. 

\noindent b) Weighted graphs:
The quantities $U_{ij}$ may represent the cost of transporting a unit of mass on that edge or may be inversely proportional to capacity of the link,  etc. The measure $\fM_U$  in this case may be far from uniform since it takes into account costs/capacities of the links. Again we can set up a maximum entropy problem with $\fM_U$  as prior obtaining a transport $\fM^*_U$ which compromises between the need to be robust and the cost/capacities of the different paths joining the source and the sink. For instance, if $U_{ij}=c_{ij}$, the cost of transporting a unit of goods on the link $ij$, is large, the solution to the maximum entropy problem with send less mass through this link provided the topology of the graph allows for alternative routes. In this case, low cost and robustness of the transportation plan may be effectively conjugated.  Indeed, we have the following striking result which generalizes Propositions \ref{uniformpath} and  \ref{uniformfeasiblepaths}.

\begin{theorem}\label{robustOMT} $\fM^*[\delta_{1x},\delta_{nx}](x)$ assigns equal probability to paths $x\in\mathcal X^{N+1}$ of equal cost. In particular, it assigns maximum and equal probability to minimum cost paths.
\end{theorem}
\begin{proof}For a path $x=(x_0,\,x_1,\ldots,x_N)$, we have
\begin{eqnarray}\nonumber\fM^*[\delta_{1x},\delta_{nx}](x)=\delta_{1x}(x_0)\frac{\varphi_v(N,x_N)}{\varphi_v(0,x_0)}\prod_{t=1}^{N-1} b_{x_t x_{t+1}}\\=\delta_{1x}(x_0)\frac{\varphi_v(N,x_N)}{\varphi_v(0,x_0)}\exp[-\sum_{t=1}^{N-1} U_{x_t x_{t+1}}].
\end{eqnarray}
Observe once more that $\delta_{1x}(x_0)\frac{\varphi_v(N,x_N)}{\varphi_v(0,x_0)}$ does not depend on the particular path joining $x_0$ and $x_N$. Since $\sum_{t=1}^{N-1} U_{x_t x_{t+1}}$ is the total cost of the path, the conclusion now follows. 
\end{proof}

In the discrete optimal mass transport (OMT) problem, one usually (e.g., see \cite{rachev1998mass}) seeks to first identify the least costly path(s) $(x_0,x_1^*,\ldots,x_{N-1}^*,x_N)$ from any starting node $x_0\in\mathcal X$ to any ending node $x_N$,
along with the corresponding end-point cost for a unit mass\footnote{We assume a self loop for each node with zero cost, i.e., $U_{xx}=0$ for each $x\in\mathcal X$.},
\[
C_{x_0x_N}=\min_{x_1^*,\ldots,x_{N-1}^*} \left(U_{x_0x_1^*}+\ldots+U_{x_{N-1}^*x_N}\right).
\]
This is a {\em combinatorial problem} but can also be cast as a linear program \cite{bazaraa2011linear}.
Having a solution to this first problem, the OMT problem can then be recast as the linear program
\begin{align}\label{eq:MKlinearprogram}
&\min_q\left\{\sum_{x_0,x_N} q_{x_0,x_N} C_{x_0x_N}\mid q_{x_0,x_N}\geq 0,\right.\\
&\hspace*{1.5cm}\left.\sum_{x_0} q_{x_0,x_N} = \nu_N(x_N),\nonumber
\sum_{x_N} q_{x_0,x_N} = \nu_0(x_0)\right\}.
\end{align}
The solution to \eqref{eq:MKlinearprogram} is the transport plan $q_{x_0,x_N}$ which dictates the portion of mass that is to be sent from $x_0$ to $x_N$ along the corresponding least costly path $(x_0,x_1^*,\ldots,x_{N-1}^*,x_N)$. Alternatively, the OMT problem can be directly cast as a linear program in as many variables as there are edges  \cite{bazaraa2011linear}.

An apparent
shortcoming of the OMT formalism is the ``rigidity'' of the transportation to utilize only paths with minimal cost
from starting to ending node. The transport provided by Theorem \ref{robustOMT}, which readily generalizes to any two marginals $\nu_0$ and $\nu_N$, provides an attractive alternative to the OMT approach: Minimum cost paths all have maximum probability, but some of the mass is also transported on alternative paths thereby ensuring a certain amount of robustness of the transportation plan. Also notice that the Schr\"odinger bridge measure $\fM^*_U[\delta_{1x},\delta_{nx}]$ determines, as a by-product, the minimum cost paths!

The argument provided at the end of the previous section (see Proposition \ref{prop3}) shows once more that $\fM^*_U[\delta_{1x},\delta_{nx}]$ can be obtained in both of the above cases in one step as solution to the Schr\"odinger bridge problem with the same marginals $\delta_{1x}$ and $\delta_{nx}$ and prior transition matrix $B$. 

All problems considered in this and in the previous section may be solved in the same way if the initial and/or the final mass is spread over several nodes.

\begin{figure}[h]
\begin{center}
\includegraphics[width=0.43\textwidth]{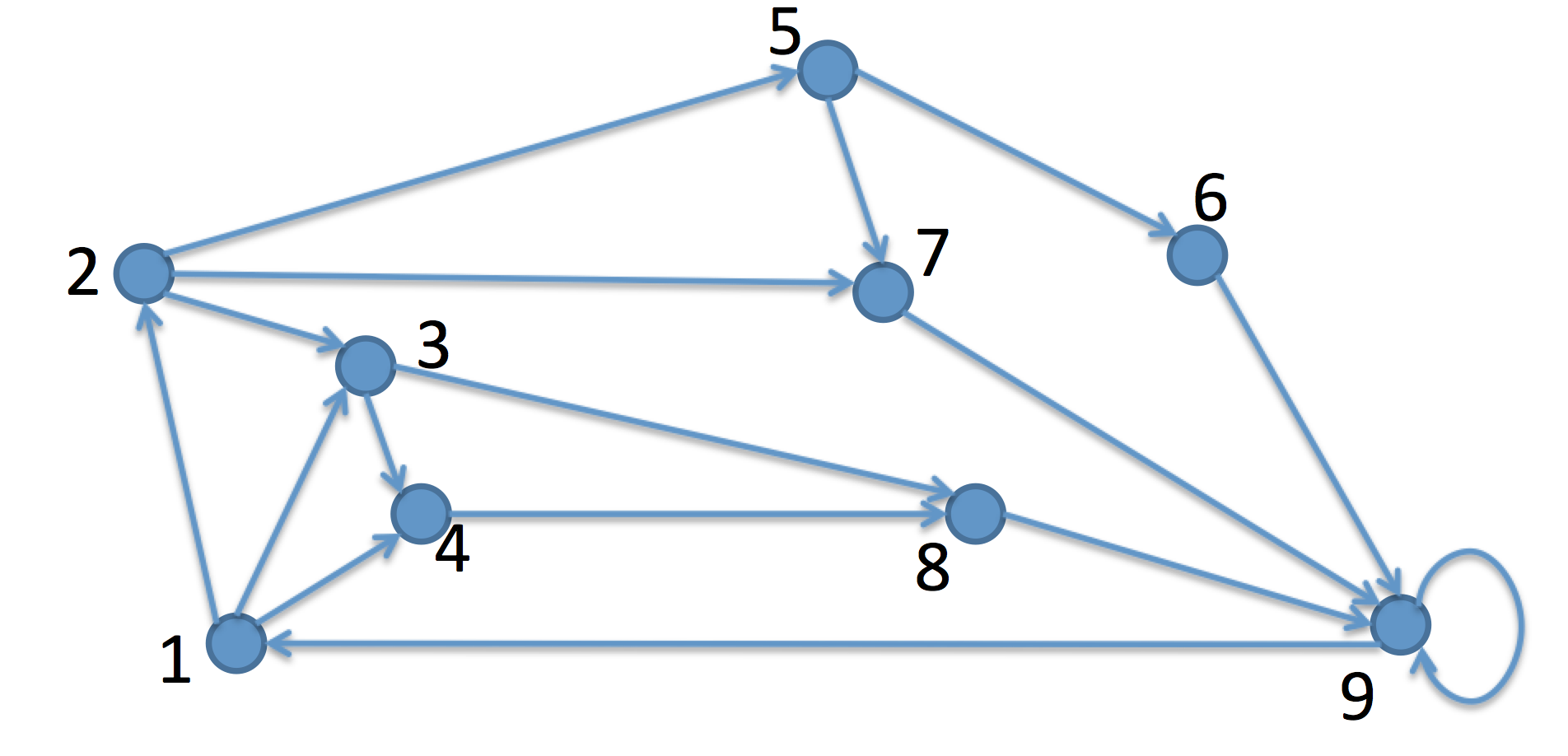}
\caption{Network topology}
\label{fig:graph}
\end{center}
\end{figure}

\section{Examples}\label{examples}
We present a simple academic example to illustrate our method. Consider the graph in Figure \ref{fig:graph} with the following adjacency matrix
 \[
       A=\tiny \left[
       \begin{matrix}
       0 & 1 & 1 & 1 & 0 & 0 & 0 & 0 & 0\\
       0 & 0 & 1 & 0 & 1 & 0 & 1 & 0 & 0\\
       0 & 0 & 0 & 1 & 0 & 0 & 0 & 1 & 0\\
       0 & 0 & 0 & 0 & 0 & 0 & 0 & 1 & 0\\
       0 & 0 & 0 & 0 & 0 & 1 & 1 & 0 & 0\\
       0 & 0 & 0 & 0 & 0 & 0 & 0 & 0 & 1\\
       0 & 0 & 0 & 0 & 0 & 0 & 0 & 0 & 1\\
       0 & 0 & 0 & 0 & 0 & 0 & 0 & 0 & 1\\
       1 & 0 & 0 & 0 & 0 & 0 & 0 & 0 & 1
       \end{matrix}
       \right].
\]
We seek to transport a unit mass from node $1$ to node $9$ in $N =3$ and $4$ steps. We add a self loop at node $9$, i.e., $a_{99}=1$, to allow for transport paths with different step sizes.

The shortest path from node $1$ to $9$ is of length $3$ and there are three such paths, which are
$1-2-7-9$, $1-3-8-9$ and $1-4-8-9$. If we want to transport the mass with minimum number of  steps, we may end up using one of these three paths. This is not so robust. On the other hand, if we apply the Schr\"odinger bridge framework with the RB measure $\fM_{\rm RB}$ as the prior, then we get a transport plan with equal probabilities using all these three paths. The evolution of mass distribution is given by
 \[
      \footnotesize
       \left[
       \begin{matrix}
       1 & 0 & 0 & 0 & 0 & 0 & 0 & 0 & 0\\
       0 & 1/3 & 1/3 & 1/3 & 0 & 0 & 0 & 0 & 0\\
       0 & 0 & 0 & 0 & 0 & 0 & 1/3 & 2/3 & 0\\
       0 & 0 & 0 & 0 & 0 & 0 & 0 & 0 & 1
       \end{matrix}
       \right],
\]
where the four rows of the matrix show the mass distribution at time step $t=0, 1, 2 ,3$ respectively. As we can see, the mass spreads out first and then goes to node $9$. When we allow for more steps $N=4$, the mass spreads even more before reassembling at node $9$, as shown below
 \[
       \footnotesize
       \left[
       \begin{matrix}
       1 & 0 & 0 & 0 & 0 & 0 & 0 & 0 & 0\\
       0 & 4/7 & 2/7 & 1/7 & 0 & 0 & 0 & 0 & 0\\
       0 & 0 & 1/7 & 1/7 & 2/7 & 0 & 1/7 & 2/7 & 0\\
       0 & 0 & 0 & 0 & 0 & 1/7 & 1/7 & 2/7 & 3/7\\
       0 & 0 & 0 & 0 & 0 & 0 & 0 & 0 & 1
       \end{matrix}
       \right].
\]

Now we change the graph by adding a cost on the edge $(7,\,9)$. In particular, we consider the weighted adjacency matrix
 \[
       B=
       \tiny
       \left[
       \begin{matrix}
       0 & 1 & 1 & 1 & 0 & 0 & 0 & 0 & 0\\
       0 & 0 & 1 & 0 & 1 & 0 & 1 & 0 & 0\\
       0 & 0 & 0 & 1 & 0 & 0 & 0 & 1 & 0\\
       0 & 0 & 0 & 0 & 0 & 0 & 0 & 1 & 0\\
       0 & 0 & 0 & 0 & 0 & 1 & 1 & 0 & 0\\
       0 & 0 & 0 & 0 & 0 & 0 & 0 & 0 & 1\\
       0 & 0 & 0 & 0 & 0 & 0 & 0 & 0 & 0.5\\
       0 & 0 & 0 & 0 & 0 & 0 & 0 & 0 & 1\\
       1 & 0 & 0 & 0 & 0 & 0 & 0 & 0 & 1
       \end{matrix}
       \right].
\]
When $N=3$ steps is allowed to transport a unit mass from node $1$ to node $9$, the evolution of mass distribution for the optimal transport plan is given by
 \[
       \footnotesize
       \left[
       \begin{matrix}
       1 & 0 & 0 & 0 & 0 & 0 & 0 & 0 & 0\\
       0 & 1/5 & 2/5 & 2/5 & 0 & 0 & 0 & 0 & 0\\
       0 & 0 & 0 & 0 & 0 & 0 & 1/5 & 4/5 & 0\\
       0 & 0 & 0 & 0 & 0 & 0 & 0 & 0 & 1
       \end{matrix}
       \right].
\]
The mass travels through paths $1-2-7-9$, $1-3-8-9$ and $1-4-8-9$, but unlike the unweighted case, the transport plan doesn't take equal probability for these three paths  Since we added a cost on the edge $(7,\,9)$, the probability that the mass takes this path becomes smaller. The plan does, however, assign equal probability to the two minimum cost paths $1-3-8-9$ and $1-4-8-9$ in agreement with Theorem \ref{robustOMT}. Suppose now we allow for more steps $N=4$ and change the $B$ matrix to
\[B=
       \tiny
       \left[
       \begin{matrix}
       0 & 0.7 & 0.7 & 0.7 & 0 & 0 & 0 & 0 & 0\\
       0 & 0 & 0.7 & 0 & 0.7 & 0 & 0.7 & 0 & 0\\
       0 & 0 & 0 & 0.7 & 0 & 0 & 0 & 0.7 & 0\\
       0 & 0 & 0 & 0 & 0 & 0 & 0 & 0.7 & 0\\
       0 & 0 & 0 & 0 & 0 & 0.7 & 0.7 & 0 & 0\\
       0 & 0 & 0 & 0 & 0 & 0 & 0 & 0 & 0.7\\
       0 & 0 & 0 & 0 & 0 & 0 & 0 & 0 & 0.5\\
       0 & 0 & 0 & 0 & 0 & 0 & 0 & 0 & 0.7\\
       0.7 & 0 & 0 & 0 & 0 & 0 & 0 & 0 & 0.9
       \end{matrix}
       \right].
\]
Here, transporting on any edge is expensive. It is, however, more expensive to transverse link $(7,9)$ and less expensive to let the mass sit at the sink node $9$. The evolution of the mass distribution is now
 \[
       \tiny
       \left[
       \begin{matrix}
         1   &     0     &    0    &     0    &     0   &      0    &     0    &     0   &      0\\
         0   & 0.5042 &   0.3173 &   0.1785    &     0     &    0     &    0    &     0   &      0\\
         0    &     0  &  0.1388  &  0.1388  &  0.2380     &    0  &  0.1275  &  0.3569    &     0\\
         0    &     0     &    0    &     0     &    0  &  0.1388  &  0.0992  &  0.2776  &  0.4844\\
         0     &    0     &    0     &    0    &     0    &     0   &      0    &     0   & 1.0000
           \end{matrix}
       \right].
\]
We observe that almost one half of the mass ($0.4844$) reaches node $9$ in three steps, and then sits there, travelling on the three shortest paths $1-2-7-9$, $1-3-8-9$ and $1-4-8-9$. As before, more mass ($0.1785$) travels on the two minimum cost paths $1-3-8-9$ and $1-4-8-9$ in agreement with Theorem \ref{robustOMT}, whereas $0.1275$ travels on the more expensive, minimum length path $1-2-7-9$. There are now several other ways the mass can reach node $9$ in $4$ steps. Our robust transportation plan takes full advantage of them, transporting more that one half of the total mass along these alternative paths.

Finally, we consider the case where the underlying graph is not strongly connected. In particular, we delete several links in Figure \ref{fig:graph} to make it not strongly connected and consider the graph in Figure \ref{fig:graph1}.
\begin{figure}[h]
\begin{center}
\includegraphics[width=0.43\textwidth]{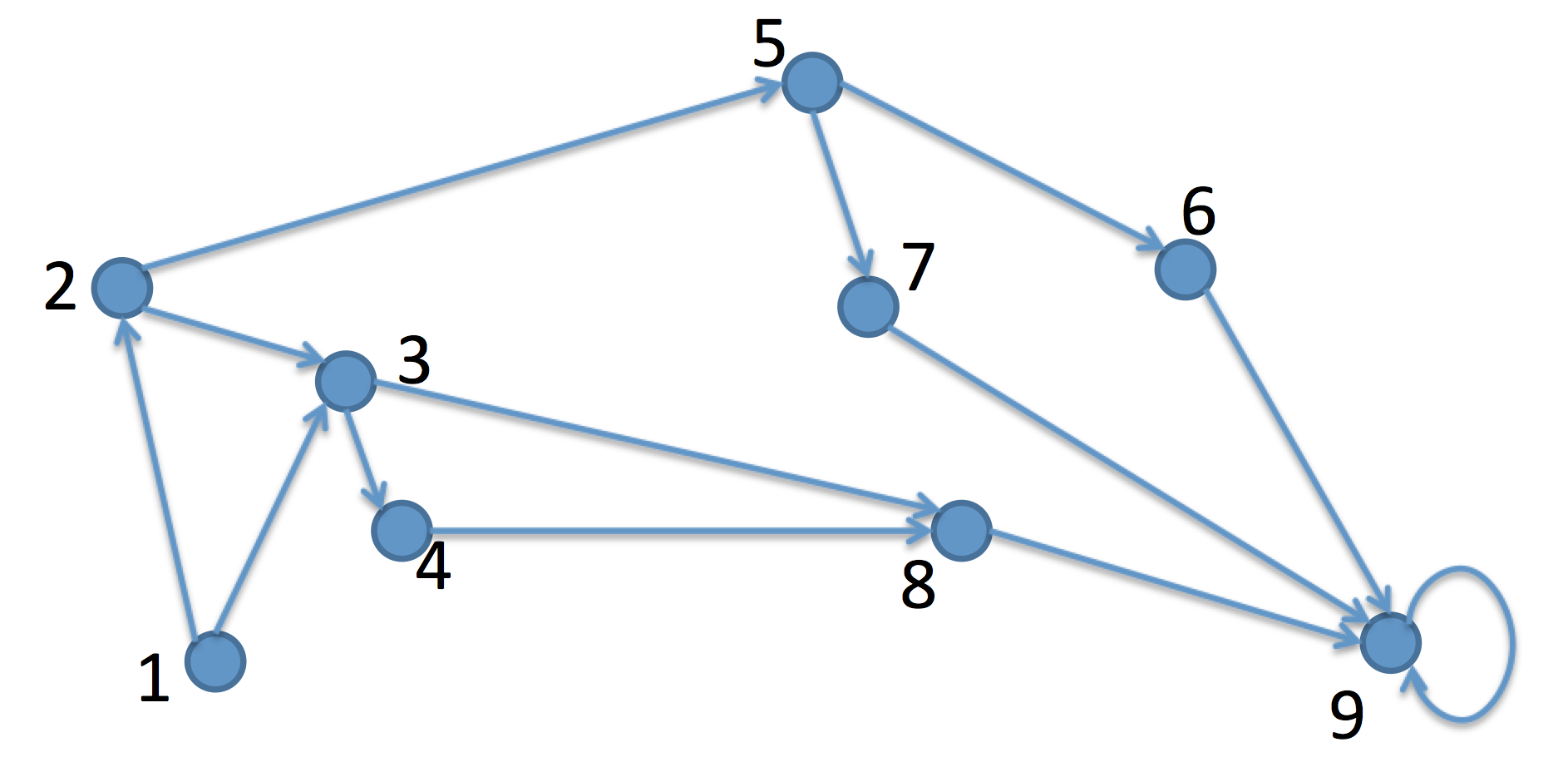}
\caption{Network topology}
\label{fig:graph1}
\end{center}
\end{figure}
Again we want to transport a unit mass from node $1$ to node $9$. In order to do this, we add an artificial energy $U_0$ to each non existing link as discussed in Section \ref{weighted}. We display the results for $N=4$ steps. When we take $U_0=2$, the evolution of mass is
 \[      \tiny
       \left[
       \begin{matrix}
       1   &      0    &     0    &     0    &     0    &     0     &    0     &    0    &     0\\
       0.0415 \!\! &\!\! 0.4079\!\! &\!\!   0.3416 \!\! &\!\!   0.0326  \!\! &\!\!  0.0462  \!\! &\!\!  0.0326  \!\! &\!\!  0.0326  \!\! &\!\!  0.0326 \!\! &\!\!  0.0326\\
       0.0270 \!\! &\!\!  0.0349 \!\! &\!\!   0.1740 \!\! &\!\!   0.1477  \!\! &\!\!  0.2330 \!\! &\!\!   0.0603 \!\! &\!\!  0.0603  \!\! &\!\!  0.1614 \!\! &\!\!   0.1014\\
       0.0116 \!\! &\!\!  0.0152 \!\! &\!\!   0.0199  \!\! &\!\!  0.0242  \!\! &\!\!  0.0163 \!\! &\!\!   0.1709 \!\! &\!\!   0.1709 \!\! &\!\!   0.2641 \!\! &\!\!   0.3069\\
         0  \!\! &\!\!       0   \!\! &\!\!      0     \!\! &\!\!    0    \!\! &\!\!     0   \!\! &\!\!      0    \!\! &\!\!     0     \!\! &\!\!    0  \!\! &\!\!  1
       \end{matrix}
       \right].
\]
We can see that there is quite a portion of mass traveling along virtual (non existing) edges. If we increase the value to $U_0=8$, then the mass evolution becomes
 \[\tiny
       \left[
       \begin{matrix}
       1.0000   \!\! &\!\!      0    \!\! &\!\!     0     \!\! &\!\!    0   \!\! &\!\!      0   \!\! &\!\!      0    \!\! &\!\!     0     \!\! &\!\!    0   \!\! &\!\!      0\\
    0.0001 \!\! &\!\!   0.5995 \!\! &\!\!   0.4000  \!\! &\!\!  0.0001 \!\! &\!\!   0.0001  \!\! &\!\!  0.0001  \!\! &\!\!  0.0001  \!\! &\!\!  0.0001 \!\! &\!\!   0.0001\\
    0.0000 \!\! &\!\!   0.0000 \!\! &\!\!   0.2000  \!\! &\!\!  0.1999   \!\! &\!\! 0.3994 \!\! &\!\!   0.0002  \!\! &\!\!  0.0002 \!\! &\!\!   0.1999 \!\! &\!\!   0.0004\\
    0.0000  \!\! &\!\!  0.0000  \!\! &\!\!  0.0000 \!\! &\!\!   0.0001 \!\! &\!\!   0.0000  \!\! &\!\!  0.1999  \!\! &\!\!  0.1999 \!\! &\!\!   0.3995\!\! &\!\!    0.2007\\
         0   \!\! &\!\!      0   \!\! &\!\!      0  \!\! &\!\!       0   \!\! &\!\!      0   \!\! &\!\!      0  \!\! &\!\!       0    \!\! &\!\!     0  \!\! &\!\!  1.0000
                \end{matrix}
       \right].
\]
The portion of mass traveling along non existing edges is negligible. Eventually, all the mass would be transported along feasible paths and in the limit the mass evolution (flow) is given by the rows of
 \[\footnotesize
       \left[
       \begin{matrix}
       1 & 0 & 0 & 0 & 0 & 0 & 0 & 0 & 0\\
       0 & 3/5 & 2/5 & 0 & 0 & 0 & 0 & 0 & 0\\
       0 & 0 & 1/5 & 1/5 & 2/5 & 0 & 0 & 1/5 & 0\\
       0 & 0 & 0 & 0 & 0 & 1/5 & 1/5 & 2/5 & 1/5\\
       0 & 0 & 0 & 0 & 0 & 0 & 0 & 0 & 1
       \end{matrix}
       \right].
\]

\section{Conclusions}
In this paper, we have proposed a novel approach to design a robust transportation plan on a given directed graph. It is based on a sort of generalized maximum entropy problem (Schr\"{o}dinger bridge) for measures on paths of the given network.  Taking as prior measure the Ruelle-Bowen-Parry random walker, the solution naturally tends to spread the mass on all available routes joining the source and the sink. Hence, the resulting transport appears robust with respect to links/nodes failure. This approach  can be adapted to 
graphs that are not strongly connected, as well as to weighted graphs. In the latter case, it can be used to effectively compromise between robustness and cost. Indeed, we exhibit a {\em robust} transportation plan which assigns {\em maximum probability} to {\em minimum cost paths} and therefore appears attractive when compared with Optimal Mass Transportation approaches. Since the transport plan is computed  as a Schr\"odinger bridge,
for which an efficient iterative algorithm is available, our procedure  also appears to be computationally attractive. 

In this paper,  in order to avoid obscuring the fundamental ideas and to keep the paper at a reasonable length, we have chosen to present the essential features of our approach without touching on a number of related fascinating topics.  For instance, in this paper {\em robustness} of a transport plan simply means that, in case of failure of certain links (e.g. due to congestion) or nodes, most of the mass will anyway reach the target nodes. There are, however,  other notions of robustness in graph theory \cite{callaway2000network,jamakovic2007relationship,wang2014wireless,sandhu2015graph,sandhu2015market},  some related to entropic principles \cite{arnold1994evolutionary,demetrius2005robustness}. 

When weights represent costs, our approach of Section \ref{weighted} compromizing  between minimization  and robustness can be further compared to {\em Optimal Mass Transport} (OMT) over graphs \cite{leonard2013lazy}, where only cost matters,  and entropically regularized OMT-schemes  \cite{cuturi2013sinkhorn,benamou2015iterative}. In discrete OMT, however, the cost function is supposed to be given, although computing it is typically an intractable problem for large networks. 

Also, it is apparent that choosing the uniform as prior distribution in the maximum entropy problem such as in Section \ref{RTG} we obtain a spreading of trajectories over which the transport occurs similar to the one in Optimal Mass Transport (OMT) on manifolds with positive Ricci-Curbastro {\em curvature} \cite{Vil08}. On discrete spaces and graphs, similar
notions of curvature have been defined by Ollivier \cite{ollivier2009ricci,ollivier2010survey}. They capture robustness and connectedness, convexity of entropy, and are related to the spectral gap \cite{bauer2011ollivier,jost2014ollivier}. Their relevance in applications is discussed in, e.g., \cite{banirazi2012heat,wang2014wireless,sandhu2015graph,sandhu2015market}. It is therefore natural to investigate the precise connection between the role of the prior in random evolutions such as those studied in this paper and deterministic evolution on discrete curved spaces.
All of these fascinating topics deserve further investigation and will be addressed elsewhere.

\spacingset{1}
\bibliographystyle{IEEEtran}
\bibliography{refs}
\end{document}